\documentclass[a4paper,
               DIV=10,
               11pt,
               titlepage=off]{scrartcl}

\usepackage{graphicx} 
\usepackage{mathpazo}
\usepackage[margin=1in]{geometry}
\usepackage{amsmath,amssymb,amsthm}
\usepackage{mathtools}
\usepackage{enumitem,microtype}
\usepackage[dvipsnames]{xcolor}
\usepackage[most]{tcolorbox}
\usepackage[colorlinks=true,linkcolor=blue!55!black,urlcolor=blue!55!black,citecolor=blue]{hyperref}

\usepackage{complexity}

\usetikzlibrary{arrows.meta,positioning,shapes.geometric,backgrounds}
 
\definecolor{cEst}{RGB}{31,88,160}    
\definecolor{cSearch}{RGB}{34,120,62} 
\definecolor{cQ}{RGB}{110,55,160}     
\definecolor{cSpec}{RGB}{190,105,20}  
\definecolor{cUnc}{RGB}{95,95,95}

\newtheorem{theorem}{Theorem}[section]
\newtheorem*{theorem*}{Theorem}
\newtheorem{corollary}{Corollary}

\newtheorem{lemma}[theorem]{Lemma}
\newtheorem{proposition}[theorem]{Proposition}
\theoremstyle{definition}
\newtheorem{definition}{Definition}
\newtheorem{assumption}{Assumption}
\theoremstyle{remark}
\newtheorem*{remark}{Remark}
\newcommand{\betti}{\beta}
\newcommand{\skel}{\mathrm{skel}}

\renewcommand{\skel}{\operatorname{skel}}
\newcommand{\rank}{\operatorname{rank}}
\newcommand{\Ker}{\operatorname{ker}}

\newcommand{\Gnp}{G(n,\tfrac12)}
\newcommand{\Gnpk}{G(n,\tfrac12,k)}

\usepackage{authblk}
\title{Average-case hardness of Betti number estimation}
\author{Sergii Strelchuk\thanks{\texttt{Sergii.Strelchuk@cs.ox.ac.uk}}}
\author{Sathyawageeswar Subramanian\thanks{\texttt{Sathya.Subramanian@cs.ox.ac.uk}}}
\author{Adam Wesołowski\thanks{\texttt{Adam.Wesolowski@cs.ox.ac.uk}}}
\affil{\small \textit{Department of Computer Science, University of Oxford, Parks Rd, Oxford OX1 3QG, United Kingdom}}
\date{}

\begin{document}
\maketitle
\vspace{-1.4cm}
\begin{abstract}
We establish the \textit{average-case} hardness of Betti number estimation on random clique complexes via a reduction from the planted clique problem. We further show that our reduction implies a series of hardness results for many problems in both classical and quantum Topological Data Analysis (qTDA).  
Under the classical planted clique conjecture, no randomized polynomial-time Betti number estimator achieves additive error below $\tfrac12$ with
constant advantage. Under a new quantum planted clique conjecture that we introduce, the
same conclusion holds for quantum polynomial-time algorithms. We also obtain related conditional hardness results for homology vanishing, additive approximations with larger error tolerances, preparation of simplex and harmonic states, cycle recovery, and counting eigenvalues at low energy.
Our reduction clarifies the structural requirements for quantum advantage in TDA and provides a new lens to investigate the classical and quantum complexity of related problems. 
\end{abstract}

\tableofcontents
\newpage

\clearpage
\section{Introduction}
\label{sec:intro}
Topological Data Analysis is a framework for studying the shape of complex
data in high dimensions. One of its basic invariants is the Betti number
$\betti_d$, which is the dimension of the $d$-th homology group and counts
independent homology classes in dimension $d$. TDA has been applied to
biological data~\cite{topaz2015topological}, voting
data~\cite{lum2013extractingwithtda}, and financial time
series~\cite{1financial}. It has also been used in cancer
research~\cite{nicolau2011topology,rabadan2020identification} and medical
imaging~\cite{singh2023topological}.
For large complexes, constructing and reducing the relevant boundary matrices
can be computationally expensive. This has motivated a long line of work on
quantum algorithms for Betti number estimation. See
Ref.~\cite{GunnKornerup2019} for an early review. Complexity results have also
identified limitations of quantum algorithms for related topological
problems~\cite{Edelsbrunner2014,SchmidhuberLloyd2023}. These worst-case results
do not, however, determine the complexity of instances that may be drawn from a specified
distribution. In this work, we study the classical and quantum
\textit{average-case} hardness of Betti number estimation under the
distribution associated with the planted clique problem.

\paragraph{Quantum algorithms for computing Betti numbers.}
Lloyd, Garnerone, and Zanardi (LGZ)\\~\cite{LloydGarneroneZanardi2016} were the first to propose a quantum algorithm that estimates the
\emph{normalized} Betti number $\beta_d/f_d$ to inverse-polynomial additive
error. A decade of follow-up work has both built this sketch
into a genuine pipeline and, simultaneously optimised the algorithms. On the constructive
side, these works have given rigorous analyses of the spectral gap and density dependencies~\cite{GunnKornerup2019,GyurikCadeDunjko2022}, extension to persistent Betti
numbers~\cite{Hayakawa2022}, qubit- and depth-efficient implementations~\cite{Ubaru2021,McArdleGilyenBerta2022,Akhalwaya2022}, and end-to-end resource
estimates for candidate advantage regimes~\cite{Berry2024}. On the limiting
side, the efficiency guarantee has been understood to rest on three promises:
\begin{enumerate}
    \item the complex must be clique-dense (else the simplex state cannot be prepared efficiently);
    \item the spectral gap of $\Delta_d$ must be at least inverse polynomial;
    \item and the normalized
    Betti number must not be exponentially small.
\end{enumerate}

Path-integral and Krylov-type classical algorithms estimate
normalized Betti numbers under closely related density and gap assumptions~\cite{Apers2022,Berry2024}, eroding the speedup exactly where the quantum
algorithm is efficient. Furthermore, Schmidhuber and Lloyd observed that generic complexes have normalized Betti numbers so small that the estimate carries no
information precisely where state preparation is easiest~\cite{SchmidhuberLloyd2023}.

\paragraph{The worst-case complexity landscape.}
In parallel, a sequence of hardness results has outlined the intrinsic difficulty of problems in qTDA. Computing Betti numbers of a clique complex exactly
is \#\P-hard, and approximating them, even to within a multiplicative
factor is
\NP-hard~\cite{SchmidhuberLloyd2023}. Computing the Euler characteristic (an
alternating sum of the number of faces (elements) of dimensions $i$ in the complex) is already \#\P-hard for general abstract
complexes~\cite{RouneSaenz2013}. Deciding whether the homology of a clique
complex vanishes is \NP-hard~\cite{AdamaszekStacho2016}. Cade and Crichigno established this for the closely
related cohomology problem of supersymmetric systems~\cite{CadeCrichigno2024},
Crichigno and Kohler for clique complexes~\cite{CrichignoKohler2024},
and King and Kohler for the gapped (promise) version on weighted graphs, which
they also place inside \QMA~\cite{KingKohler2023}. For the estimation task the
sharpest calibration is due to Gyurik, Cade, and Dunjko~\cite{GyurikCadeDunjko2022}: estimating the low-lying spectral density of a
succinctly specified sparse Hermitian matrix is DQC$_1$-hard, so an efficient classical
simulation of the quantum estimation pipeline on all inputs would collapse the one-clean-qubit
model. On the related problem of \textit{homology localization}, Chen and Freedman proved that computing
minimal cycle representatives of homology classes is \NP-hard to approximate~\cite{ChenFreedman2011}.
 
All of these results are \emph{worst-case}: the hard instances are adversarial, and none of them addresses the complexity of Betti-number
estimation on the inputs TDA is typically run on. For instance, the theory of random
topology shows that clique complexes of Erd\H{o}s--R\'enyi graphs have highly
structured homology: for $X(G(n,p))$ for a fixed $0<p<1$ the nontrivial homology statistically concentrates in a narrow
band of dimensions and vanishes outside it~\cite{Kahle2009,Kahle2014}, so the measure of the set of worst-case instances may simply be extremely close to zero on random inputs. Whether any natural distribution over inputs
makes Betti-number estimation hard \textit{on average} for classical and
quantum algorithms alike has, to our knowledge, remained open. It is this
gap that we address in this paper.

\paragraph{Average-case hardness from planted clique.}
\label{sec: avg-case-hardness}
Our source of average-case hardness is the \textit{planted clique problem}, introduced by
Jerrum~\cite{Jerrum1992} and Ku\v cera~\cite{Kucera1995}: distinguish an
Erd\H{o}s--R\'enyi random graph $G(n,\tfrac12)$ from the same graph with a clique
planted on $k$ random vertices that we call $\Gnpk$. Spectral and combinatorial algorithms succeed
in polynomial time for $k = \Omega(\sqrt n)$~\cite{AlonKrivelevichSudakov1998,DekelGurelGurevichPeres2014,DeshpandeMontanari2015},
while for $(2+\varepsilon)\log_2 n \le k \le n^{1/2-\varepsilon}$ no efficient
algorithm is known despite three decades of effort. The conjecture that none
exists is supported by \textit{unconditional lower bounds} in every standard restricted
model: statistical-query algorithms~\cite{FeldmanGRVX2017}, the sum-of-squares
hierarchy~\cite{BarakHKKMP2019}, Lov\'asz--Schrijver relaxations~\cite{FeigeKrauthgamer2003}, bounded-depth circuits~\cite{Rossman2008}, the
Metropolis process~\cite{Jerrum1992}, and low-degree polynomial tests~\cite{KuniskyWeinBandeira2019,SchrammWein2022}. On this basis, the planted
clique conjecture (PC) has become a standard hardness primitive for
average-case reductions, underlying lower bounds for sparse PCA~\cite{BerthetRigollet2013}, submatrix detection~\cite{MaWu2015}, community
detection~\cite{HajekWuXu2015}, certification of restricted isometry~\cite{KoiranZouzias2014}, Nash-equilibrium approximation~\cite{HazanKrauthgamer2011}, and a web of further problems~\cite{BrennanBreslerHuleihel2018,BrennanBresler2020,ManurangsiRubinsteinSchramm2021}.
Recently, Hirahara and Shimizu proved that many formalizations of PC are essentially equivalent~\cite{HiraharaShimizu2024}. Their
equivalences supply two ingredients we use throughout this work. First, hardness at
\emph{every} constant distinguishing advantage is equivalent to the standard
conjecture, which lets us state our theorem for algorithms succeeding with
probability $\tfrac12+\delta$ for arbitrarily small constant $\delta>0$.
Second, their results calibrate exactly how far such statements can be pushed:
a folklore edge-counting test distinguishes the two distributions with
advantage $\Theta(k^2/n)$~\cite{Trevisan2018,HiraharaShimizu2024}, so no reduction can rule out all
inverse-polynomial advantage, a subtlety that has invalidated at least one
published hardness framework~\cite{ElrazikRobereSchusterYehuda2022} as noted by~\cite{HiraharaShimizu2024}. Therefore, our constant-advantage statement for (q)TDA is also optimal in this
sense. In~\cite{HiraharaShimizu2024} the authors further identify a detection--recovery gap~\cite{HiraharaShimizu2024,BreslerJiang2023,SchrammWein2022}: while detection
admits the $\Theta(k^2/n)$-advantage edge count, \emph{recovering} the clique
is hard under PC conjecture even with success probability $n^{-c}$ for every constant
$c$, a distinction our corollaries inherit. Figure~\ref{fig:Implications} depicts all implications we have identified as corollaries of our main result. In particular, state preparation for the LGZ algorithm, homology localization, vanishing homology, Betti number to large additive approximation, are all hard on average under the distribution $\mathcal{D}_k$ originating from Erd\H{o}s--R\'enyi clique complexes.

\paragraph{The importance of the distribution choice.}

Average-case hardness is inherently defined over some fixed choice of a distribution $D(x)$ of problem instances $x$. Different choices of $D(x)$ might lead to different \textit{average-case} hardness statements, thus a choice of $D(x)$ should be well motivated. In this work we adopt the distribution $\mathcal{D}_k \;=\; \tfrac12\,G(n,\tfrac12)\;+\;\tfrac12\,G(n,\tfrac12,k)$ from the planted clique problem, which naturally derives from the Erd\H{o}s--R\'enyi distribution $G(n,\tfrac12)$ (see sec~\ref{sec:graphsandmodels} for definitions). Erd\H{o}s--R\'enyi distribution is a very common and well studied distribution of random complexes~\cite{Kahle2009,Kahle2014,sharma2018solving, BIO_ER} which has been shown to be relevant to modeling the human brain~\cite{BIO_ER}, as well as to machine learning and computer vision tasks such as object matching, object detection, and Structure from Motion (SfM)~\cite{sharma2018solving}.
 We use
the balanced distribution for the planted clique problem. This choice is useful
because planted clique is widely used in average-case reductions and because
our reduction leaves the graph instance unchanged. We do
not claim that $\mathcal{D}_k$ is a statistical model for the data usually
encountered in TDA. It remains open to establish comparable hardness results
for distributions arising in
applications.

\begin{figure}[htpb!]
    \centering

\begin{tikzpicture}[
  font=\scriptsize,
  every node/.style={align=center},
  spoke/.style={-{Latex[length=2.2mm]}, line width=0.7pt, draw=#1!75!black, shorten >=2pt, shorten <=2pt},
  chain/.style={-{Latex[length=2mm]}, line width=0.6pt, densely dashed, draw=black!60, shorten >=2pt, shorten <=2pt},
  box/.style={rounded corners=2pt, draw=#1!80!black, fill=#1!12, line width=0.7pt,
              text width=3.05cm, inner sep=4pt, minimum height=1.05cm},scale=0.97
  ]

\node[ellipse, draw=black, line width=1.1pt, fill=black!6, inner sep=5pt,
      text width=5.2cm] (THM) at (0,0)
  {\textbf{\small Average-case hardness of Betti-number estimation}\\[2pt]
   Under the classical and quantum \textit{PC} conjectures~\ref{ass:pc_classical} and~\ref{ass:pc}, no poly-time algorithm estimates
   $\beta_d\bigl(\mathrm{skel}_d X(G)\bigr)$, $d\in\mathcal B$, to additive error $<\tfrac12$
   over $\mathcal{D}_k \;=\; \tfrac12\,G(n,\tfrac12)\;+\;\tfrac12\,G(n,\tfrac12,k)$.\\[2pt]
   \emph{The reduction:} $\;\beta_d=0$ (null) vs.\ $\beta_d\ge\binom{k-1}{d+1}$ (planted), w.h.p.};
 
\newcommand{\place}[4]{\node[box=#2] (#1) at ({7.35*cos(#3)},{5.15*sin(#3)}) {#4};}
 
\place{APPROX}{cEst}{90}{\textbf{Multiplicative \& coarse additive approximation}\\
  any factor; additive error up to $e^{c\log^2 n}$. (cf.~\cite{SchmidhuberLloyd2023})}
\place{NORM}{cEst}{57}{\textbf{Normalized Betti number estimation}\\
  $\beta_d/f_d$: $\,0$ vs.\ $1-\tfrac{d+1}{k}$ the quantity used in LGZ algo.}
\place{LLSD}{cSpec}{20}{\textbf{Low-lying spectral density}\\
  $\#\{\lambda_i(\Delta_d)<\lambda\}$, $\lambda\in(0,k)$; avg-case of the DQC1-hard LLSD (cf.~\cite{GyurikCadeDunjko2022}).}
\place{HSTATE}{cQ}{-15}{\textbf{Harmonic-state preparation}\\
  mixed state on $\ker\Delta_d$ (QTDA output) to const.\ trace distance}
\place{HARM}{cSearch}{-55}{\textbf{Harmonic representatives}\\
  any vector near $\ker\Delta_d$; kernel is $1$-dim.\ on the plant at $d=k-2$}
\place{SSTATE}{cQ}{-92}{\textbf{Simplex-state preparation}\\
  $|\psi_d\rangle\propto\sum_{\sigma\in X_d}|\sigma\rangle$ (LGZ input state)}
\place{SAMP}{cSearch}{-131}{\textbf{Simplex search \& sampling}\\
  output/sample any $d$-simplex; the ``clique sampling'' subroutine}
\place{LOC}{cSearch}{-168}{\textbf{Homology localization}\\
  any nonzero $d$-cycle / minimal representative (cf.~\cite{ChenFreedman2011}).}
\place{VANISH}{cEst}{160}{\textbf{Homology vanishing}\\
  decide $H_d=0$; avg-case analogue of \QMA$_1$-hard clique homology (cf.~\cite{CrichignoKohler2024,KingKohler2023, SchmidhuberLloyd2023})}
\place{UNC}{cUnc}{124}{\textbf{Unconditional lower bounds}\\
  no low-degree test, no SQ algorithm, no poly-size SoS certificate of $H_d=0$ (see section~\ref{sec: avg-case-hardness} for references)}
 
\draw[spoke=cEst]    (THM) -- (APPROX);
\draw[spoke=cEst]    (THM) -- (NORM);
\draw[spoke=cSpec]   (THM) -- (LLSD);

\draw[spoke=cQ]      (THM) -- (HSTATE);
\draw[spoke=cSearch] (THM) -- (HARM);
\draw[spoke=cQ]      (THM) -- (SSTATE);
\draw[spoke=cSearch] (THM) -- (SAMP);
\draw[spoke=cSearch] (THM) -- (LOC);
\draw[spoke=cEst]    (THM) -- (VANISH);
\draw[spoke=cUnc, densely dotted] (THM) -- (UNC);

\draw[chain] (SAMP) to[bend right=12] node[midway, below, sloped, font=\tiny]{} (SSTATE);
\draw[chain] (HARM) to[bend right=12] node[midway, below, sloped, font=\tiny]{} (HSTATE);

\end{tikzpicture}

    \caption{Illustration of selected implications of the reduction from \textit{planted clique problem} to Betti number estimation. The references in each box link the worst-case hardness variant corresponding to our average-case results. To the best of the authors knowledge the remaining problems do not have an established worst-case hardness.}
    \label{fig:Implications}
\end{figure}
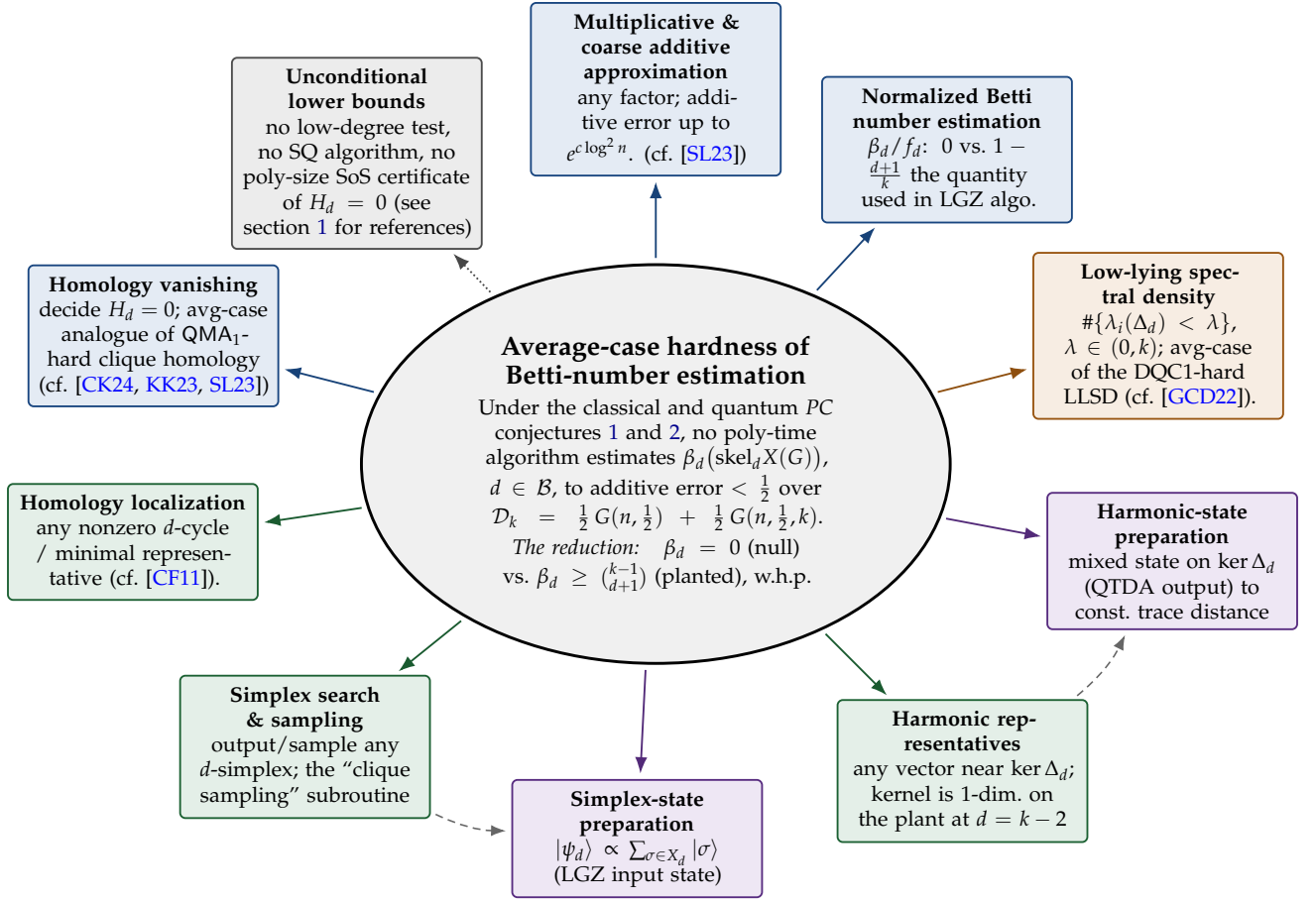

\paragraph{A new lens on quantum advantage in TDA.}
Our results carry implications in two directions for the prospect of quantum advantage.
Negatively, they close off one tempting argument: Betti number estimation and other aforementioned qTDA tasks remain hard on natural random inputs, so under this distribution one cannot witness a
quantum speedup and moreover,
at the critical dimension, even the pipeline's state-preparation step is
infeasible, for reasons independent of the well-known density obstruction~\cite{GyurikCadeDunjko2022,SchmidhuberLloyd2023,Berry2024}. Positively, they
sharpen the search: any demonstration of quantum advantage in TDA must locate
instance families that evade our reduction. For example, complexes that are
clique-dense with inverse-polynomially gapped Laplacians and large normalized
Betti numbers and remain classically hard for reasons \emph{other} than the planted clique
structure. We show that under the distribution of instances coming from the planted clique problem the problem of Betti number estimation is hard for both classical and quantum
algorithms. 
The reduction thus functions as a filter on proposed advantage
regimes, and clarifies the boundary between the worst-case quantum hardness
results~\cite{CadeCrichigno2024,CrichignoKohler2024,KingKohler2023,GyurikCadeDunjko2022}
and the average-case landscape closer to the practical data on which TDA is expected to operate.

\subsection{Organisation}
\label{sec:org}

The remainder of the paper is organised as follows. Section~\ref{sec:contributions} introduces our
main theorem, the average-case hardness of Betti-number estimation on clique
complexes drawn from the graph distribution of the planted clique problem. Section~\ref{sec:background} fixes notation and recalls the definitions of
planted-clique problem, clique complexes and Betti numbers, and the local
access model through which the graph specifies the complex and its
combinatorial Laplacian in polynomial time. Section~\ref{sec:RED} gives the reduction
from planted clique to Betti-number estimation, proves Theorem~\ref{thm:main}, and
discusses how the argument behaves for skeleta of other dimensions and random graphs with varying parameter $p$.
Section~\ref{sec:algIMPL} draws out the algorithmic consequences. Section~\ref{sec:identifyadvantage} outlines how our contribution helps with the search for quantum advantage in qTDA. Section~\ref{sec:implications} derives the
corollaries for other tasks of quantum Topological Data Analysis,
namely normalized Betti number estimation, homology vanishing, homology
localization, simplex- and harmonic-state preparation and the low-lying
spectral density, and compares them with worst-case state-preparation
hardness. Section~\ref{sec:open} closes with a discussion and a list of open problems.

\section{Contributions}
\label{sec:contributions}
Our results connect the hardness of the well known problem in graph theory, namely the planted clique problem, to a family of problems in computational algebraic topology, known broadly under the name of topological data analysis.
We give the first conditional average-case hardness reduction for problems in topological data
analysis. Our work provides a new framework that allows one to investigate questions in algebraic topology via graph theoretical tools.
Based on the established \textit{classical planted clique conjecture}~\cite{HiraharaShimizu2024, Kucera1995, Jerrum1992} we derive the classical average-case hardness results.
We further introduce \textit{the quantum planted clique conjecture}    (see~\ref{ass:pc}), based on which we derive the quantum conditional average-case hardness of Betti number estimation (and other problems) of clique complexes drawn from the planted clique distribution
$\mathcal{D}_k = \tfrac12 G(n,\tfrac12) + \tfrac12 G(n,\tfrac12,k)$. A demonstration of an exponential quantum speedup for the Betti number estimation under $\mathcal{D}_k$ would not only falsify the quantum planted clique conjecture but also imply the first average-case quantum speedup (assuming the classical planted clique conjecture stands).

In section~\ref{sec:RED} we prove the following:
\begin{theorem}[Average-case hardness of Betti number estimation]
\label{thm:main}
Let $k=k(n)$ satisfy the inequality $(2+\epsilon)\log{n} < k<n^{1/2-\epsilon}$ and set $d=k-2$. For every
constant $\delta>0$, the following statements hold.
Under Assumption~\ref{ass:pc_classical}, no randomized polynomial-time
algorithm estimates
\[
\betti_d\bigl(\skel_d X(G)\bigr)
\]
to additive error below $\tfrac12$ with success probability at least
$\tfrac12+\delta$ over $G\sim\mathcal{D}_k$.
Under Assumption~\ref{ass:pc}, no quantum polynomial-time algorithm with a
classical output estimate satisfies the same guarantee.
\end{theorem}
See section~\ref{sec:background} for definitions of $\textsc{Betti}$ and section~\ref{sec:RED} for the statement of planted-clique conjectures.
In section~\ref{sec:implications} we extend the scope of the average-case hardness results to a family of problems in topological data analysis and other problems related to Betti number estimation, the extent of results is depicted in Figure~\ref{fig:Implications}.

\subsection{Technical overview}
We give a simple reduction from \textit{planted clique problem} to the
Betti-number estimation on clique complexes. 
The starting observation is that the clique
complex of a complete graph on $k$ vertices is a $k-1$ simplex $\sigma^{k-1}$,
which is contractible, so planting a $k$-clique contributes no non-trivial homology classes.
The reduction proceeds by fixing the dimension $k$ in the PC window (see~
\ref{ass:pc}) and
truncating the complex at dimension $d = k-2$, the planted $k$-clique a
topologically invisible full simplex  becomes the boundary sphere
$\partial\sigma^{k-1} \cong S^{k-2}$, contributing a non-trivial homology class while the null complex has, with high probability, no
$(k{-}2)$-simplex at all. Consequently $\beta_{k-2}$ equals $0$ in the null
case and is at least $1$ in the planted case, and any algorithm estimating it
to additive error below $\tfrac12$ decides planted clique. Our main theorem
states that, under the classical and quantum PC conjectures, no polynomial-time classical or quantum algorithm
achieves this with success probability $\tfrac12+\delta$ for any constant
$\delta>0$ over $\mathcal{D}_k=\frac{1}{2}\Gnp+\frac{1}{2}\Gnpk$. Notably, the
input to the Betti problem \emph{is} the random graph itself and the reduction
is the identity on the instance so the hardness attaches to the same
input distribution. We also present our result as complementing the
worst-case \#\P-, \NP-, \QMA$_1$-, and DQC$_1$-hardness results 
above~\cite{SchmidhuberLloyd2023,AdamaszekStacho2016,CrichignoKohler2024,GyurikCadeDunjko2022}
with an average-case statement, and complementing random-topology theory~\cite{Kahle2009,Kahle2014} with a computational hardness statement.

\subsection{Classical and quantum results}
We highlight that the reduction presented in section~\ref{sec:RED} is the same both classically and in the quantum setting. Under the well established classical planted clique conjecture (Assumption~\ref{ass:pc_classical}) Betti number estimation and other problems mentioned in section~\ref{sec:implications} are \textit{classically} average-case hard. We introduce a quantum clique conjecture (assumption~\ref{ass:pc}) under which the conditional average-case hardness follows in the quantum setting. A polynomial time quantum algorithm for Betti number estimation in the hard regime would render the conjecture false and the problems in section \ref{sec:implications} quantum-average-case easy, while remaining classical average-case hard. This remains an open problem.

\section{Background}
\label{sec:background}

\subsection{Notation}
\label{sec:notation}
 
Unless stated otherwise, all asymptotics are taken as the number of vertices
$n \to \infty$, and $\log$ denotes $\log_2$. We write $[n] = \{1,\dots,n\}$ and
$\binom{[n]}{j}$ for the family of $j$-element subsets of $[n]$. We use the
standard Landau symbols $O(\cdot)$, $o(\cdot)$, $\Omega(\cdot)$, $\Theta(\cdot)$,
and write $f \lesssim g$ as shorthand for $f = O(g)$. An event
$E = E_n$ holds \emph{with high probability} (w.h.p.) if $\Pr[E_n] \to 1$ as
$n \to \infty$. Two fixed positive constants recur throughout: $\varepsilon > 0$ is used to
delimit the window of the planted-clique conjecture, and
$\delta > 0$ denotes a distinguishing advantage. We say a randomised or quantum
algorithm is \emph{efficient} if it runs in time $\poly(n)$.

\paragraph{Graphs and random models.}
\label{sec:graphsandmodels}
$G$ denotes a simple graph on vertex set $[n]$. The Erd\H{o}s--R\'enyi
distribution $G(n,\tfrac12)$ includes each of the $\binom{n}{2}$ edges
independently with probability $\tfrac12$; the \emph{planted} distribution
$G(n,\tfrac12,k)$ is obtained from $G_0 \sim G(n,\tfrac12)$ by choosing a
uniformly random $k$-subset $S \subseteq [n]$ and adding every edge inside $S$.
$\mathrm{PClique}(n,k)$ is the associated detection problem
(Definition~\ref{def:PC}), and $C_{\max}(G)$ the size of a largest clique
of $G$.
 \paragraph{Simplicial and clique complexes.}
An (abstract) simplicial complex $K$ on vertex set $[n]$ is a
family of subsets of $[n]$ such that any  $\sigma, \tau$ satisfy: $\sigma \in K$ and $\tau \subseteq \sigma$ then
$\tau \in K$. A set $\sigma \in K$ of cardinality $d+1$ is a \emph{$d$-simplex},
and its subsets of cardinality $d$ are its $d-1$-\emph{faces} or \textit{facets} of $\sigma$. We write $K_d$ for the
set of $d$-simplices of $K$ and $f_d(K) = |K_d|$ for their number.
 
For a graph $G$ on $[n]$, the \emph{clique complex} (or \emph{flag complex})
$X(G)$ is the simplicial complex whose $d$-simplices are exactly the
$(d{+}1)$-cliques of $G$; equivalently, a set of vertices is a simplex iff it is
pairwise adjacent. Thus $X(G)$ is determined by the graph alone, and we identify
a $d$-simplex with the $(d{+}1)$-subset of $[n]$ that spans it. The
\emph{$d$-skeleton} $\skel_d X(G)$ retains every face of dimension $\le d$ and
discards all faces of dimension $> d$. Writing $f_d(G) := f_d(X(G))$ for the
number of $d$-simplices (equivalently, of $(d{+}1)$-cliques), independence of the
$\binom{n}{2}$ edges gives (in the Erd\H{o}s--R\'enyi model with $p=1/2$)
\begin{equation}
  \mathbb{E} f_d(G)
  \;=\; \binom{n}{d+1}\, 2^{-\binom{d+1}{2}},
  \label{eq:expected-faces}
\end{equation}
since a $(d{+}1)$-clique requires all $\binom{d+1}{2}$ of its edges to be
present.

\paragraph{Chains, boundaries, and homology.}

We work over $\mathbb{R}$ throughout. A $j$-chain is a formal real linear
combination of the oriented $j$-simplices of $K$, and the $j$-chains form a
vector space $C_j=C_j(K;\mathbb{R})$. The \emph{boundary map}
$\partial_j : C_j \to C_{j-1}$ sends each $j$-simplex to the signed sum of the
$(j{-}1)$-simplices on its boundary; it records the oriented
boundary of a chain. Its defining property is that \emph{the boundary of a
boundary is empty}, $\partial_{j-1}\partial_j = 0$.
 
Two subgroups of $C_j$ carry the topological content. The group of
\emph{$j$-cycles}
\[
  Z_j \;=\; \ker \partial_j \;=\; \{\, c \in C_j : \partial_j c = 0 \,\}
\]
consists of the closed chains, those with empty boundary. The group of
\emph{$j$-boundaries}
\[
  B_j \;=\; im \partial_{j+1} \;=\; \{\, \partial_{j+1} c : c \in C_{j+1} \,\}
\]
consists of the chains that are themselves the boundary of some
$(j{+}1)$-chain, that is, the cycles already filled in by
$(j{+}1)$-simplices sitting one dimension above. Since the boundary of a boundary
is empty, every boundary is a cycle, so $B_j \subseteq Z_j$. Homology records the
cycles that are left over: the closed $j$-cycles that do \emph{not} bound any
higher-dimensional part of the complex. The \emph{$j$-th homology} and
\emph{Betti number} are
\begin{equation}
  H_j(K) \;=\; Z_j / B_j,
  \qquad
  \beta_j(K) \;=\; \dim H_j(K),
\end{equation}
so $\beta_j$ counts the independent cycles that do not bound any higher dimensional structures. Concretely
$\beta_0$ is the number of connected components, $\beta_1$ the number of
independent loops, and $\beta_d$ the number of independent $d$-dimensional
voids. Real coefficients agree with rational ones and with the integral rank,
$\beta_j(K)=\operatorname{rank}_{\mathbb{Z}}H_j(K;\mathbb{Z})$; over $\mathbb{F}_p$ the
dimension of $H_j$ can be strictly larger, by the torsion contribution.

\paragraph{The combinatorial Laplacian and Hodge theory.}
The \emph{combinatorial Laplacian} in dimension $d$ is
\begin{equation}
  \Delta_d
  \;=\; \partial_d^{\top}\partial_d + \partial_{d+1}\partial_{d+1}^{\top},
  \label{eq:laplacian}
\end{equation}
a positive-semidefinite operator on $C_d$. The combinatorial Hodge theorem
identifies harmonic chains with homology classes, so that for every finite
complex $K$,
\begin{equation}
  \beta_d(K) \;=\; \dim H_d(K) \;=\; \dim \ker \Delta_d(K).
  \label{eq:hodge}
\end{equation}
This spectral characterisation is what makes Betti numbers amenable to
Laplacian-based (in particular quantum) algorithms. We will work almost
exclusively with a \emph{top} skeleton $K = \skel_d X(G)$, in which there are no
$(d{+}1)$-simplices, so $\partial_{d+1} = 0$ and the second (``filling'') term
of \eqref{eq:laplacian} vanishes. Then $\Delta_d = \partial_d^{\top}\partial_d$
and \eqref{eq:hodge} collapses to the elementary identity
\begin{equation}
  \beta_d\!\left(\skel_d X(G)\right)
  \;=\; \dim \ker \partial_d
  \;=\; f_d(G) - \rank \partial_d.
  \label{eq:top-skeleton}
\end{equation}

\subsection{Random graphs and the planted-clique problem}

$\Gnp$ denotes the Erd\H{o}s--R\'enyi graph on vertex set $[n]=\{1,\dots,n\}$ in which each
of the $\binom n2$ edges is present independently with probability $\tfrac12$. The
\emph{planted} distribution $\Gnpk$ is obtained by drawing $G_0\sim\Gnp$, choosing a
uniformly random $k$-subset $S\subseteq[n]$, and adding every edge within $S$.

\begin{definition}[Planted-clique detection]
\label{def:PC}
$\textsc{PClique}(n,k)$ is the task of deciding, given a random graph $G$ drawn from $\Gnp$ (the \emph{null} case) whether or not it has been modified by planting a $k$-clique resulting in $\Gnpk$ (the
\emph{planted} case). In other words, the problem is to decide which case it is (whether the graph comes from the distribution $\Gnp$ or not). 
\end{definition}

The largest clique size of $\Gnp$ is $C_{max}(\Gnp)=(2+o(1))\log_2 n$ w.h.p. For
$k<(2-\varepsilon)\log_2 n$ the plant is information-theoretically undetectable; for
$k\gtrsim\sqrt n$ spectral algorithms detect it in polynomial time. The intermediate window
is believed intractable.

\subsection{Input and access model}

This subsection makes precise the sense in which ``estimating a Betti number of $X(G)$''
is a problem whose input is the Erd\H{o}s--R\'enyi random graph.

\begin{definition}[Betti number estimation]\label{def:problem}
$\textsc{Betti}$ takes as input a graph $G$ on $[n]$ vertices and an
integer $d$, and asks for an estimate of
\[
\betti_d\bigl(\skel_d X(G)\bigr)
\]
with coefficients in $\mathbb{R}$. The exact version asks for the value of
the Betti number. An additive approximation with error $\eta>0$ outputs
$\widehat{\betti}$ satisfying
\[
|\widehat{\betti}-\betti_d|<\eta.
\]
For $C\ge1$, a $C$-multiplicative approximation outputs
$\widehat{\betti}$ satisfying
\[
\frac{\betti_d}{C}
\le
\widehat{\betti}
\le
C\betti_d.
\]
When $\betti_d=0$, this definition requires
$\widehat{\betti}=0$. We write $\textsc{Betti}(d)$ when the dimension is
fixed.
\end{definition}

The complex $\skel_d X(G)$ and its Laplacian $\Delta_d$ have up to $f_d=n^{\Theta(d)}$
simplices. The Laplacian and the complex are \emph{specified implicitly} by $G$
through a local oracle:

\begin{itemize}[leftmargin=1.4em,itemsep=1pt]
\item a $d$-simplex is named by a $(d{+}1)$-subset of $[n]$, a label of $O(d\log n)$ bits;
\item testing whether a label is a simplex is $\binom{d+1}{2}$ edge look-ups in $G$;
\item the non-zero entries of any row of $\Delta_d=\partial_d^{\!\top}\partial_d$ correspond
      to $d$-simplices sharing a $(d{-}1)$-face with the given one, enumerable by local
      clique checks among $O(dn)$ candidates.
\end{itemize}

Thus from $G$ alone one can answer any local query about $\Delta_d$ in $\mathrm{poly}(n)$
time, this is a standard and efficient access model, and one of the reasons people study clique complexes. The graph is a polynomial-size exact description of the implicit complex; \emph{constructing the Laplacian $\Delta_d$ explicitly} is exponential but is not required,
whereas querying it locally is efficient. Consequently the only conceivable
super-polynomial cost lies in extracting the global spectral quantity
$\dim\Ker\Delta_d$ from this local description which is precisely the difficulty of the Betti number estimation.

 \section{Reduction from PC to Betti number estimation}
\label{sec:RED}
\begin{assumption}[Classical planted-clique conjecture, \textup{PC}]\label{ass:pc_classical}
There is a range of plant sizes

$$
  (2+\varepsilon)\log_2 n \;\le\; k \;\le\; n^{1/2-\varepsilon}
$$

in which \textit{no randomized polynomial-time classical algorithm} distinguishes $\Gnp$ from $\Gnpk$ with
advantage $\ge\delta$, for any constant $\delta>0$. 
\end{assumption}
\begin{assumption}[Quantum planted-clique conjecture, \textup{qPC}]\label{ass:pc}
There is a range of plant sizes

$$
  (2+\varepsilon)\log_2 n \;\le\; k \;\le\; n^{1/2-\varepsilon}
$$

in which \textit{no polynomial-time quantum algorithm} distinguishes $\Gnp$ from $\Gnpk$ with
advantage $\ge\delta$, for any constant $\delta>0$. 
\end{assumption}

This constant-advantage form is
equivalent to the standard formulations of the conjecture by the equivalences of
Hirahara--Shimizu~\cite{HiraharaShimizu2024}, so nothing stronger than the usual assumption is being made. It is a standard detection form of the conjecture, it has been used to show
the \textit{average-case} hardness of sparse PCA, community detection, and related problems, and is
supported by unconditional lower bounds against the Sum-of-Squares, statistical-query, and
low-degree-polynomial models.

Fix $k$ in the PC window \ref{ass:pc} and set $d=k-2$. We show estimating
$\beta_{k-2}\bigl(\skel_{k-2}X(G)\bigr)$ w.h.p. decides the
\textit{planted clique problem}.

\begin{lemma}[Null value]\label{lem:null}
For $k\ge(2+\varepsilon)\log_2 n$, w.h.p.\ $G_0\sim\Gnp$ has no $(k{-}1)$-clique, whence
$\beta_{k-2}\bigl(\skel_{k-2}X(G_0)\bigr)=0$.
\end{lemma}

\begin{proof}
The expected number of $(k{-}1)$-cliques is
$\binom{n}{k-1}2^{-\binom{k-1}{2}}\le\bigl(n\,2^{-(k-2)/2}\bigr)^{k-1}$. For
$k\ge(2+\varepsilon)\log_2 n$ the base is $n^{-\varepsilon/2+o(1)}\to0$, so by Markov's
inequality no $(k{-}1)$-clique exists w.h.p. Then $\skel_{k-2}X(G_0)$ has no
$(k{-}2)$-simplex at all: $f_{k-2}=0$, and $\beta_{k-2}=0$.
\end{proof}

\begin{lemma}[Planted value]\label{lem:plant}
For $k\ge(2+\varepsilon)\log_2 n$, w.h.p.\ over $\Gnpk$,
$\beta_{k-2}\bigl(\skel_{k-2}X(G)\bigr) \geq 1$.
\end{lemma}

\begin{proof}
The planted $k$-clique is a complete graph, corresponding to a clique complex that is the full simplex
$\sigma^{k-1}$. Truncating at dimension $k-2$ deletes its unique top $(k{-}1)$-face and
leaves the boundary $\partial\sigma^{k-1}\cong S^{k-2}$, a simplicial $(k{-}2)$-sphere with
$\beta_{k-2}=1$. By Lemma~\ref{lem:null}, w.h.p.\ the original graph likely contributes no
$(k{-}2)$-simplex outside the planted structure, so w.h.p the only $(k{-}2)$-simplices in $\skel_{k-2}X(G)$
are the $k$ facets of the plant. These alone support exactly one nontrivial homology class, contributing
$\beta_{k-2}=\dim\Ker\partial_{k-2}= 1$. Nevertheless, it is possible that edges from the plant complete another $k-clique$ in the graph resulting in more than one non-trivial homology class in dimension $d=k-2$, i.e. $\betti_d\geq1$ in the planted case.
\end{proof}
\begin{remark}
    For the sake of the reduction it is sufficient that the Betti number in dimension $d$ increases to 1 in the planted case, even though it might increase to a value higher than 1.
\end{remark}

\begin{lemma}[Deciding PC via Betti number]\label{lem:indicator}
For $k\ge(2+\varepsilon)\log_2 n$, w.h.p.\ over the distrubution defining $\textsc{PClique}(n,k)$,
$$
  \beta_{k-2}\bigl(\skel_{k-2}X(G)\bigr)
  \;\geq\;
  1\iff[\,G\ \text{is from the planted case}\,].
$$
$$
  \beta_{k-2}\bigl(\skel_{k-2}X(G)\bigr)
  \;=\;
  0\iff[\,G\ \text{is from the null case}\,].
$$
\end{lemma}

\begin{proof}
Immediate from Lemmas~\ref{lem:null} and~\ref{lem:plant}: the value is $0$ in the null case
and equal to $1$ in the planted case, each w.h.p.
\end{proof}

\begin{theorem*}~\ref{thm:main} [Average-case hardness of Betti number estimation]

Let $k=k(n)$ satisfy the inequality $(2+\epsilon)\log{n} < k<n^{1/2-\epsilon}$, and set $d=k-2$. For every
constant $\delta>0$, the following statements hold.
Under Assumption~\ref{ass:pc_classical}, no randomized polynomial-time
algorithm estimates
\[
\betti_d\bigl(\skel_d X(G)\bigr)
\]
to additive error below $\tfrac12$ with success probability at least
$\tfrac12+\delta$ over $G\sim\mathcal{D}_k$.
Under Assumption~\ref{ass:pc}, no quantum polynomial-time algorithm with a
classical output estimate satisfies the same guarantee.
\end{theorem*}

\begin{proof}
Suppose an efficient $\mathcal{A}$ solves $\textsc{Betti}(k-2)$ to additive error $<\tfrac12$
with success probability $\ge \tfrac12+\delta$ over $D_k$. Define a detector for
$\textsc{PClique}(n,k)$: on input $G$, form $skel_{k-2}X(G)$, run $\mathcal{A}$, and output
\textsc{planted} iff $\mathcal{A}(skel_{k-2}X(G))\ge \tfrac12$.

Write $s_0=\Pr_{G\sim G(n,\frac12)}[\mathcal{A}\text{ succeeds}]$ and
$s_1=\Pr_{G\sim G(n,\frac12,k)}[\mathcal{A}\text{ succeeds}]$. Since $D_k$ is the balanced
mixture, the hypothesis says $\tfrac12(s_0+s_1)\ge\tfrac12+\delta$, i.e.
\begin{equation}\label{eq:sum-form}
  s_0+s_1 \;\ge\; 1+2\delta .
\end{equation}
Note that \eqref{eq:sum-form} bounds neither $s_0$ nor
$s_1$ individually. Let $\eta=\eta(n)=o(1)$ bound the probability of the bad events of
Lemmas~\ref{lem:null} and~\ref{lem:plant}.

\emph{Planted case.} Except with probability $\eta$ we have $\beta_{k-2}\ge 1$; on the success
event of $\mathcal{A}$ the additive error $<\tfrac12$ forces $\mathcal{A}(G,k-2)>\tfrac12$, so
\[
  \Pr_{G\sim G(n,\frac12,k)}\bigl[\text{output \textsc{planted}}\bigr]\;\ge\; s_1-\eta .
\]

\emph{Null case.} Except with probability $\eta$ we have $\beta_{k-2}=0$; on the success event
the estimate lies below $\tfrac12$, so the detector outputs \textsc{null}. Hence
\[
  \Pr_{G\sim G(n,\frac12)}\bigl[\text{output \textsc{planted}}\bigr]\;\le\; (1-s_0)+\eta .
\]

Subtracting and applying \eqref{eq:sum-form}, the advantage is at least
\[
  \bigl(s_1-\eta\bigr)-\bigl(1-s_0+\eta\bigr)\;=\;(s_0+s_1)-1-2\eta\;\ge\;2\delta-2\eta\;\ge\;\delta
\]
for all large $n$, since $\delta>0$ is constant and $\eta=o(1)$. This contradicts
Assumptions~\ref{ass:pc_classical} and~\ref{ass:pc} in the classical and quantum settings respectively.
\end{proof}
\begin{remark}[Why the distribution is hard]
Neither component of $\mathcal{D}_k$ is hard alone: by Lemma~\ref{lem:null} the constant algorithm
``output $0$'' solves $\mathrm{BETTI}(k-2)$ exactly with probability $1-o(1)$ over
$G(n,\tfrac12)$, and by Lemma~\ref{lem:plant} ``output $1$'' does so over
$G(n,\tfrac12,k)$. The hardness is that no single algorithm can handle the input distribution $\mathcal{D}_k$. 
\end{remark}

\begin{remark}[Optimality of the constant-advantage form for the decision problem.]
The constant $\delta>0$ in Theorem~\ref{thm:main} cannot be improved to an inverse
polynomial: counting edges and thresholding yields an efficient algorithm that estimates
$\beta_{k-2}$ to additive error $<\tfrac12$ with success probability
$\tfrac12+\Theta(k^2/n)$ on each distribution, and $k^2/n=n^{-2\alpha}$ for
$k=n^{1/2-\alpha}$. Thus Theorem~\ref{thm:main} is stated at the strongest
success-probability threshold consistent with this trivial distinguisher; ruling out
advantage $\gg k^2/n$ is possible under the optimal decision form of \textup{PC}, but
ruling out all inverse-polynomial advantage is false.
\end{remark}

\begin{remark}[Relation to worst-case results]
Deciding whether $H_d$ of a clique complex vanishes is \QMA$_1$-hard, and computing Betti
numbers is \#\P-hard (\NP-hard for approximations), in the worst case. Theorem~\ref{thm:main} is a new \textit{average-case} result on the distribution $\mathcal{D}_k$.
\end{remark}

\subsection{Reduction for different values of \texorpdfstring{$d$}{d}}
In the previous section we have described a reduction which looks at the $d$-skeleton for a particular choice of $d$, namely $d=k-2$. In principle, we can use skeleta of the simplical complex in other dimensions too. The argument, however, becomes more nuanced. The closer $d$ is to the bottom of the hard PC range, i.e. $d\rightarrow (2+\epsilon)\log{n}$, the Betti number signal $\betti_d$ from the planted clique becomes stronger because a d-skeleton of a $k$-clique, adds $\binom{k-1}{d+1}$ non-trivial homology classes in dimension $d$. However, even though the signal for the choice of smaller $d$ increases, the background noise of the non-planted $d-cliques$ coming from the $\Gnpk$ distribution also increases as $d$ decreases. Even though for the complexity theoretic purposes the reduction carries through for a simple choice of $d=k-2$, it might be instructive to look at the mapping for different choices of the skeleton dimension.

Markov's inequality assures us that when we look at the $d$-skeleton of the complex $X(G)$ then the probability of a $C=d-1$ clique existing in the graph $\Gnp$ is \begin{equation}
    P(C>a)\leq \frac{2\log{n}}{a}
\end{equation} 
thus for any $d$-clique such that $\frac{d}{\log{n}}\rightarrow \infty$, we can set $a=d$ and have an assurance that no $d$-clique exists w.h.p. (probability $1$ as $n$ goes to infinity). Thus the reduction of theorem~\ref{thm:main} applies for every such $d$, and the planted Betti number signal (change in the number of distinct homology classes in dimension $d$) gets stronger at smaller $d$.  
\subsection{Reduction for arbitrary constant \texorpdfstring{$p$}{p}}
The average-case planted clique hardness holds for every constant value of parameter $p$~\cite{HiraharaShimizu2024}.
Thus for any distribution $\mathcal{D}_k$ with $p=const$, the average-case hardness of planted clique transfers to Betti number estimation on clique complexes from $\mathcal{D}_k$.

\section{Algorithmic implications}
\label{sec:algIMPL}
Because we have established the hardness of Betti number estimation on random clique complexes, it is natural to ask
which of the ingredients that make quantum algorithms efficient are present, and
which are not. 
We organise this section around the three
promises under which the LGZ quantum algorithm is efficient, recalled from the
introduction:
\begin{enumerate}
  \item[(P1)] the complex is clique-dense, so the uniform simplex state
        $\lvert\psi_d\rangle\propto\sum_{\sigma\in X_d(G)}\lvert\sigma\rangle$ can be prepared
        in polynomial time;
  \item[(P2)] the combinatorial Laplacian $\Delta_d$ has an at least
        polynomially small spectral gap, so $1/\mathrm{gap}=\mathrm{poly}(n)$;
  \item[(P3)] the normalized Betti number $\beta_d/f_d$ is not exponentially
        small.
\end{enumerate}

An advantage requires a family that is quantumly easy while
classically hard, hence, by the previous paragraphs and by~\cite{Apers2022},
a family on which (P1), (P2) and (P3) hold simultaneously at a growing
dimension, since at bounded constant dimension the classical estimator already matches
the quantum one. The Erd\H{o}s--R\'enyi model cannot supply this: at growing
dimension its clique density is small, (P1) fails (see proof of \ref{lem:null}), and
neither algorithm is efficient. Any candidate family must therefore lie outside the distribution $\mathcal{D}_k$, and it must in addition be classically hard for
a reason other than a hidden planted clique, because the planted structure, we
have shown, is equally opaque to quantum algorithms. Concretely, a candidate
family should satisfy the promises P1,P2,P3 and be classically hard.
The gapped clique-homology instances on weighted graphs of King and
Kohler~\cite{KingKohler2023}, which are QMA$_1$-hard, are contained in QMA and come with
an inverse-polynomial promise gap, are the natural structural template for
P2, P3 and classical hardness; what is missing is a natural distribution over such instances that satisfies P1. In this light Theorem~\ref{thm:main} functions as a
filter on proposed advantage regimes: it removes the entire
family $\mathcal{D}_k \;=\; \tfrac12\,G(n,p)\;+\;\tfrac12\,G(n,p,k)$ for any constant $p$ and, with it, any average-case advantage argument that would draw its classical
hardness from planted-clique-type hidden structure. It remains entirely plausible that there exists a natural family of instances which exhibits quantum advantage, and the results obtained in this work are a stepping stone that clarifies the complexity landscape and makes it clearer what kind of instances might exhibit quantum advantage.

\section{Identifying regions with average-case quantum advantage}
\label{sec:identifyadvantage}

For the clique complex $X(G(n,p))$ with $p=n^{-\alpha}$, Kahle's theorems\cite{Kahle2009} place
the $d$-th homology in a narrow band of densities:
\begin{equation}
\beta_d\bigl(X(G(n,p))\bigr)>0 \ \text{w.h.p.}
\qquad\Longleftrightarrow\qquad
\tfrac{1}{\,d+1\,}<\alpha<\tfrac1d ,
\label{eq:kahle-window}
\end{equation}
with $\beta_d=0$ below the window (too sparse to support $d$-cycles) and above
it (dense enough that $(d+1)$-faces fill every $d$-cycle). Two features of the
window matter here.
In particular, the claim from \cite{Kahle2009} is even stronger, namely in the non-trivial region, the non-trivial homology is close to maximal, i.e. $\frac{\betti_d}{f_d}=1$ w.h.p. Even more interestingly, we can show the non-trivial homology regime falls outside of the average-case hard region and our reduction breaks down also due to almost maximal Betti number in the Kahle window (no planted signal can be detected), possibly indicating a quantum advantage regime (subject to no efficient classical algorithms).

\subsection{Non-trivial homology in random clique-complexes}
\label{sec:nontrivhomo}

Write $p=n^{-\alpha}$ and let $X(G(n,p))$ be the random clique (flag) complex,
with $f_d$ the number of its $d$-simplices, so
$\mathbb{E}f_d=\binom{n}{d+1}p^{\binom{d+1}{2}}$. Kahle's theorems locate the
homology of $X(G(n,p))$ in a single band of densities, and show it is empty
outside that band.
 
\begin{theorem}[Kahle, concentration of homology in one dimension \cite{Kahle2009}]
\label{thm:kahle-window}
For any $d\ge 1$ and let $p=n^{-\alpha}$ with
\[
\frac{1}{d+1}<\alpha<\frac1d .
\]
Then w.h.p.\ $\widetilde H_i\bigl(X(G(n,p))\bigr)=0$ for every $i\ne d$, and
$\beta_d\bigl(X(G(n,p))\bigr)>0$.
\end{theorem}

\begin{theorem}[Kahle, vanishing outside the window \cite{Kahle2009,Kahle2014}]
\label{thm:kahle-vanish}
For any $d\ge 1$ and let $p=n^{-\alpha}$. If $\alpha>\frac1d$ (too sparse
to support $d$-cycles) or $\alpha<\frac{1}{d+1}$ (dense enough that
$(d{+}1)$-faces fill every $d$-cycle), then $\beta_d\bigl(X(G(n,p))\bigr)=0$
w.h.p.
\end{theorem}
 
Theorems~\ref{thm:kahle-window} and~\ref{thm:kahle-vanish} give the equivalence
\begin{equation}
\beta_d\bigl(X(G(n,p))\bigr)>0 \ \text{w.h.p.}
\qquad\Longleftrightarrow\qquad
\tfrac{1}{\,d+1\,}<\alpha<\tfrac1d .
\label{eq:kahle-result}
\end{equation}
Inside the window the statement is quantitatively stronger: the homology is not
merely non-zero but close to maximal, a sharpening already implicit in Kahle's
Euler-characteristic computation, which admits a short direct proof~\cite{Kahle2014}.
 
\begin{lemma}[Near-maximal homology in the window (Theorem 3.8 from \cite{Kahle2009})]
\label{lem:max-betti}
If $\alpha$ satisfies \eqref{eq:kahle-window}, then w.h.p.
\begin{equation}
\beta_d\bigl(X(G(n,p))\bigr)=(1-o(1))f_d.
\label{eq:max-betti}
\end{equation}
\end{lemma}

Thus in the window $f_d$ dominates its neighbouring face counts, so almost every
$d$-simplex supports a non-trivial homology class: the normalized Betti number is
$1-o(1)$.

We are also going to show that the reduction presented in section \ref{sec:RED} does not hold in the Kahle's window. It is therefore the natural place
to ask whether an average-case quantum advantage can be demonstrated in that window, since the average-case hardness reduction breaks down. 

\subsection{No reduction in the non-trivial homology window}
\label{sec:no-reduction}

The reduction of Section~\ref{sec:RED} decides planted clique by reading a
$0$-versus-$(\ge 1)$ gap in $\beta_d$, and that gap is manufactured by
\emph{truncating} the complex at $d = k-2$: truncation turns the contractible
planted simplex $\sigma^{k-1}$ into the boundary sphere
$\partial\sigma^{k-1}\cong S^{k-2}$, whose non-trivial homology contributes the signal. Inside Kahle's
window neither ingredient is available. First, the window lives at dimension
$d \approx 1/\alpha$, which for constant $p$ is $\Theta(\log_{1/p} n)$, slightly but sharply
\emph{below} the lower threshold of the hard window of planted clique ($(2+o(1))\log_{1/p} n$). So the empty-null case that the reduction exploits does not occur
here. Secondly, the null-case Betti number is already
near-maximal\cite{Kahle2009}. The following proposition shows in a simple way why the reduction of theorem\ref{thm:main} does not apply to the non-trivial homology region identified by Kahle.

\begin{proposition}[Kahle's window lies below the PC-hard range]
\label{prop:dimension-gap}
Let $p=n^{-\alpha}$ and let $d$ satisfy $\tfrac{1}{d+1}<\alpha<\tfrac1d$, $d>1$. Then
$d+2 < (2+\varepsilon)\log_{1/p} n$ for every $\varepsilon>0$ and $n$ large. Hence no $k$
in the PC window \ref{ass:pc} satisfies $k=d+2$, and Lemma~\ref{lem:null} fails
at every dimension in the window.
\end{proposition}

\begin{proof}
$\log_{1/p}n = 1/\alpha$, so the PC window requires $k \ge (2+\varepsilon)/\alpha$, while
$\alpha>1/(d+1)$ gives $d+2 < 1/\alpha + 2$.
\end{proof}

\subsection{Region for average-case quantum-advantage}
\label{sec:advantage-region}

Our results, together with the classical algorithm of
\cite{Apers2022} and the efficient state preparation problem\cite{GyurikCadeDunjko2022,SchmidhuberLloyd2023}, partition
the $(d,\alpha)$ landscape of Erd\H{o}s--R\'enyi clique complexes into three
regions, and by elimination we arrive at a single candidate for average-case
quantum advantage.
\begin{itemize}
  \item \textbf{PC-hard region} ($(2+\epsilon)\log{n} < d<n^{1/2-\epsilon}$ ). The null complex is empty w.h.p, $\beta_d = 0$, and the planted clique
        creates the only homology. This is where Theorem~\ref{thm:main} rules out average-case advantage:
        Betti estimation is PC-hard, but it is hard for both classical and
        quantum algorithms and, by Schmidhuber--Lloyd \cite{SchmidhuberLloyd2023}, the invariant
        carries no information on typical inputs. Thus no quantum advantage.
  \item \textbf{Bounded dimension} ($d = O(1)$). The window is nonempty here only
        for $\alpha = \Theta(1)$, and there the normalized Betti number is
        estimated in polynomial time \emph{classically} \cite{Apers2022}. Any
        quantum algorithm is matched by a classical one. Thus quantum advantage is ruled out. 
  \item \textbf{The window at growing dimension}
        ($\tfrac1{d+1} < \alpha < \tfrac1d$, $d\to\infty$). Here the normalized
        Betti number is near-maximal (P3), the gap can be favourable (P2), and,
        by ~\ref{prop:dimension-gap}, the planted-clique reduction
         provably fails. One remaining obstacle is (P1) which causes the quantum state preparation step to be inefficient with known techniques. Nevertheless, the main observation is that the average-case hardness we establish for the problem does not reach into this region. This is the interesting and non-trivial regime our reduction does not affect and the question of average-case quantum advantage remains open and plausible.
\end{itemize}
\section{Implications
}
\label{sec:implications}
In this section we present the corollaries of our reduction from section~\ref{sec:RED}. We work with the quantum hardness assumptions for planted clique (assumptions~\ref{ass:pc},~\ref{ass:SearchStrong}), but all corollaries hold for classical computation if the quantum conjectures are replaced with the more established classical hardness assumptions.

We introduce the following quantum version of the strong search PC conjecture, for the classical version we refer the reader to~\cite{HiraharaShimizu2024}. This will be of use for some of the corollaries.
\begin{assumption}[Quantum version of the strong search conjecture~\cite{HiraharaShimizu2024}]
\label{ass:SearchStrong}
For any constants $\alpha \in (0, 1/2)$ and
$c > 0$, any (classical randomized~\cite{HiraharaShimizu2024})  quantum polynomial-time algorithm fails to find
a $k$-clique in $\Gnpk$ with probability $1-n^{-c}$ (the probability of success is inverse polynomial) for all sufficiently large
$n$ and $k\leq n^{1/2-\alpha}$.
    
\end{assumption}
Usually the worst-case complexity results refer to a random clique complex $X(G)$; average-case results from this work are about the $d$-skeleton of a clique complex from distribution $\mathcal{D}_k$ in the same local-access model.
\subsection{Simple corollaries about simplicial complexes}

\begin{corollary}[Simplex search]\label{cor:simplex}
    
Under the strong search \textup{PC} conjecture~\ref{ass:SearchStrong}, for every inverse polynomial $\delta$ no
efficient algorithm outputs, with probability $\ge \delta$ over
$G\sim\Gnpk$, a single $(k{-}2)$-simplex of $X(G)$.
\end{corollary}
 
\begin{proof}
A $(k{-}2)$-simplex is a $(k{-}1)$-clique; by Lemma~\ref{lem:plant} any
such clique is w.h.p.\ contained in the plant, so its vertices are $k-1$ of
the $k$ plant vertices, verifiable by $\binom{k-1}{2}$ edge queries. This
recovers a large clique, contradicting strong search PC.
\end{proof}
 
\begin{remark}[This is not a Betti-number corollary]
Corollary~\ref{cor:simplexstate} uses only that the null complex is empty in
dimension $k-2$; it makes no reference to homology and is a consequence of the planted clique conjecture directly. We include it because it is the input
subroutine of the quantum pipeline (below), but we flag that its hardness does
not route through Theorem~\ref{thm:main}. It also holds \emph{only} at the
critical dimension: for $d\ll k$ the null complex is dense in dimension $d$
and sampling a $d$-simplex is trivial in both cases.
\end{remark}
 
\begin{corollary}[Simplex-state preparation]\label{cor:simplexstate}
Under the strong search version of \textup{PC}, no efficient quantum algorithm
prepares, to constant trace distance and with probability $\ge n^{-c}$ over
$G\sim\Gnpk$, the uniform simplex state
$$
  |\psi_{k-2}\rangle \;\propto\; \sum_{\sigma\in X_{k-2}(G)} |\sigma\rangle .
$$
\end{corollary}
 
\begin{proof}
If $\rho$ is within trace distance $\tfrac14$ of
$|\psi_{k-2}\rangle\langle\psi_{k-2}|$, then measuring $\rho$ in the
computational basis returns a $(k{-}2)$-simplex with probability $\ge\tfrac12$
.
\end{proof}
 
\begin{remark}[Implications for the LGZ pipeline]\label{rem:lgz-pipeline}
$|\psi_{k-2}\rangle$ is exactly the input state of the LGZ algorithm~\cite{LloydGarneroneZanardi2016}: the pipeline first prepares the uniform
superposition over $d$-simplices, then applies phase estimation to
$\Delta_d$. Corollary~\ref{cor:simplexstate} shows that at $d=k-2$ this first
step is already \textit{average-case} hard on the $\Gnpk$ family. One clarification is essential to state the implication correctly.
This is distinct from the well-known \emph{density} obstruction to LGZ~\cite{GyurikCadeDunjko2022,Berry2024}: LGZ preparation costs
$O((f_d/\binom{n}{d+1}^{-1/2})$ by amplitude amplification, which is
inefficient when the clique density is small. That obstruction
is unconditional; ours is conditional on the PC conjecture. At $d=k-2$ both
obstructions are present but they are not the same statement: density says
``preparation is slow because simplices are rare,'' whereas
Corollary~\ref{cor:simplexstate} says ``preparation is hard on average because deciding
whether any simplex exists is planted-clique-hard.'' 
\end{remark}

\subsection{Corollaries in Topological Data Analysis}

\begin{figure}[h!]
    \centering
    \includegraphics[width=\linewidth]{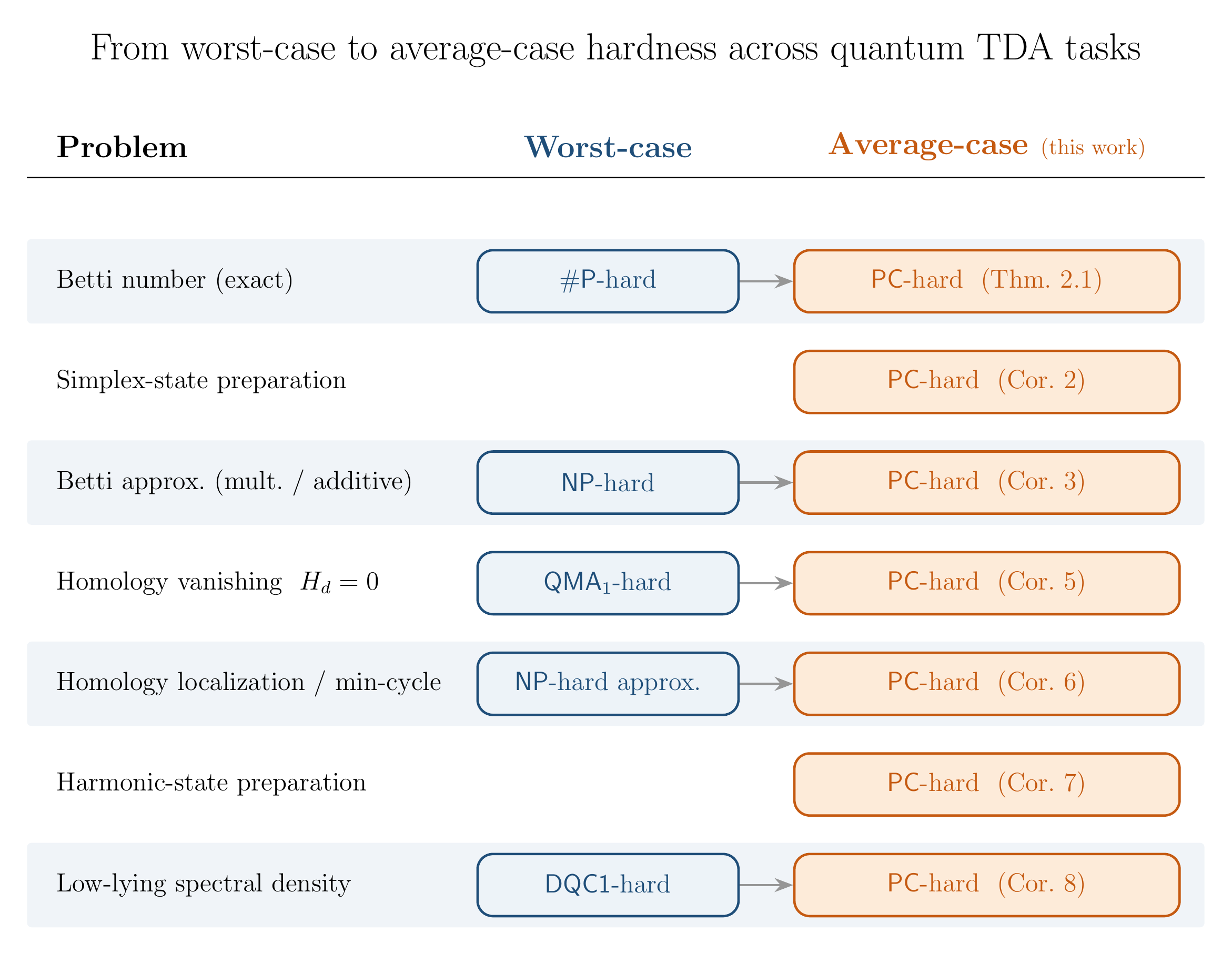}
    \caption{Illustration of comparison of the worst-case vs average-case hardness of various problems related to Topological Data Analysis. The hardness of exact and approximate Betti number computations follows from~\cite{SchmidhuberLloyd2023}, $\QMA_1$-hardness follows from ~\cite{CrichignoKohler2024}(see remark~\ref{rem:qma1_vsavg}), $\textsc{DQC}_1$-hardness from~\cite{GyurikCadeDunjko2022}, $\NP$-hardness of approximate homology localization from~\cite{ChenFreedman2011}.}
    \label{fig:COMPARISON}
\end{figure}

\begin{corollary}[Multiplicative and coarse additive approximation]\label{cor:approx}
Under Assumption~\ref{ass:pc}, no efficient algorithm outputs any finite
multiplicative approximation of $\betti_d$ with success probability
$\ge\tfrac12+\delta$ for a constant $\delta>0$. Moreover, at $d=k-2$ no
efficient algorithm estimates $\betti_{k-2}$ to additive error $<\tfrac12$
with such probability; and for $d=\Theta(\log n)$ with $k=n^{1/2-\varepsilon}$
the additive error may be relaxed to $\exp(c\log^2 n)$ for a constant $c>0$
and the conclusion still holds.
\end{corollary}
 
\begin{proof}
A multiplicative approximation distinguishes $\betti_d=0$ from
$\betti_d\ge 1$: any finite factor maps $0$ to $0$ and a positive integer to a
positive number. By Lemma~\ref{lem:indicator} this decides the planted case,
so Theorem~\ref{thm:main} applies. For the additive claim at $d=k-2$ the value
gap is exactly $1$, giving error tolerance $<\tfrac12$. For $d\ll k$,
Lemma~\ref{lem:plant} gives planted value
$\ge\binom{k-1}{d+1}$ while the null value is $0$ w.h.p.; with
$d+1>(2+\epsilon)\log{n}$ and $k=n^{1/2-\varepsilon}$ one has
$\binom{k-1}{d+1}=\exp(\Theta(\log^2 n))$, so any additive approximation to
error $<\tfrac12\binom{k-1}{d+1}=\exp(c\log^2 n)$ still separates the two
cases.
\end{proof}

\begin{corollary}[Normalized Betti-number estimation]
\label{cor:normalized}
Under \textup{PC}, fix $k$ in the window \ref{ass:pc} and any $d$ with
$(2+\varepsilon)\log_2 n \le d+1 = o(k)$. Then no efficient algorithm estimates
$\beta_d\bigl(\mathrm{skel}_d X(G)\bigr)/f_d(G)$ to additive
error $<\tfrac12$ with success probability $\ge \tfrac12+\delta$, for any
constant $\delta>0$, over $G\sim G(n,\tfrac12)$ vs. $G(n,\tfrac12,k)$.
\end{corollary}
 
\begin{proof}
Null: since $d+1\ge(2+\varepsilon)\log_2 n$, Lemma~\ref{lem:null} (with $d+1$ for
$k-1$) gives $\beta_d,f_d=0$ w.h.p., so $\beta_d/f_d=0$. We set $\beta_d/f_d := 0$ when $f_d = 0$. 
Planted: the plant contributes
$\binom{k-1}{d+1}$ independent nonbounding $d$-cycles (its $d$-skeleton is a wedge
of that many $d$-spheres), which survive in $\ker\partial_d$, so
$\beta_d\ge\binom{k-1}{d+1}$; and the expected number of $(d{+}1)$-cliques outside
the plant is $\sum_{j\le d}\binom{k}{j}\binom{n-k}{d+1-j}2^{-\binom{d+1}{2}+\binom{j}{2}}
=o\!\bigl(\binom{k}{d+1}\bigr)$ (dominated by the $j=d$ term
$\binom{k}{d+1}\tfrac{2(d+1)}{n^{1+\varepsilon}}$, using $2^d\ge\tfrac12 n^{2+\varepsilon}$),
so $f_d=(1+o(1))\binom{k}{d+1}$ w.h.p. Hence
$\beta_d/f_d\ge(1-o(1))\bigl(1-\tfrac{d+1}{k}\bigr)=1-o(1)$. The two values $0$ and
$1-o(1)$ are separated by $>\tfrac12$, so thresholding an error at $\varepsilon<\tfrac12$  decides $\mathrm{PClique}$ as in Theorem~\ref{thm:main}.
\end{proof}
 
\begin{remark}[Relation to the LGZ target precision]\label{rem:lgz-precision}
$\betti_d/f_d$ is exactly the quantity the LGZ algorithm and its descendants
estimate~\cite{LloydGarneroneZanardi2016,GyurikCadeDunjko2022}, and they aim
for inverse-polynomial additive error. Corollary~\ref{cor:normalized} shows
hardness already at additive error $<\tfrac12$. The apparent paradox with the algorithms' efficiency
is resolved by their promises: LGZ is efficient only when the normalized Betti
number is not small and the simplex state is efficiently preparable, and our
planted family violates the second promise (see
Corollary~\ref{cor:simplexstate}). 
\end{remark}

\begin{corollary}[Homology vanishing]\label{cor:vanishing}
Under Assumption~\ref{ass:pc}, no efficient algorithm decides whether
$H_d(\skel_d X(G))=0$ with success probability $\ge\tfrac12+\delta$ for a
constant $\delta>0$. 
\end{corollary}
 
\begin{proof}
The null and planted values of $\betti_d$ are $0$ and $\ge 1$ w.h.p.: the boundary sphere $\partial\sigma^{k-1}\cong S^{k-2}$ has
$\betti_{k-2}\geq1$ 
(Lemma~\ref{lem:plant}). Deciding $H_d=0$ therefore decides the planted case, and Theorem~\ref{thm:main} applies.
\end{proof}

\begin{corollary}[Homology localization and cycle representatives]\label{cor:localization}
Under the strong search version (not merely deciding whether the clique was planted) of \textup{PC}~\cite{HiraharaShimizu2024}, for every constant
$c>0$ no efficient algorithm outputs, with probability $\ge n^{-c}$ over
$G\sim\Gnpk$, any nonzero $(k{-}2)$-cycle of $\skel_{k-2}X(G)$. Equivalently
any representative of a nontrivial class in $H_{k-2}$, exact or within any
$\ell_2$-distance $<\tfrac1{\sqrt k}$ of a genuine cycle. In particular
computing minimal-weight cycle representatives is \textit{average-case} hard.
\end{corollary}
 
\begin{proof}
By Lemma~\ref{lem:plant} every nonzero cycle is a multiple of $z$, whose
support is all $k$ facets of the plant. Any vector within $\ell_2$-distance
$<\tfrac1{\sqrt k}$ of the normalised cycle $z/\sqrt k$ has all $k$ facet
coordinates nonzero (each coordinate of $z/\sqrt k$ has magnitude
$1/\sqrt k$), so reading off its support and verifying each facet by
$\binom{k-1}{2}$ edge queries recovers $k-1$ plant vertices. This contradicts
the strong search version of PC as in
Corollary~\ref{cor:simplex}. Minimal-weight representatives are a special
case of Corollary~\ref{cor:simplex}., giving an \textit{average-case} analogue of the worst-case hardness of~\cite{ChenFreedman2011}.
\end{proof}

\begin{corollary}[Harmonic-state preparation]\label{cor:harmonicstate}
Under the strong search version of \textup{PC}, no efficient quantum algorithm
prepares, to constant trace distance and with probability $\ge n^{-c}$ over
$G\sim\Gnpk$, the maximally mixed state
$\Pi_0/\betti_{k-2}$ on $\ker\Delta_{k-2}$.
\end{corollary}
 
\begin{proof}
By Lemma~\ref{lem:plant} the kernel is at least one-dimensional at $d=k-2$, so
$\Pi_0/\betti_{k-2}=|h_G\rangle\langle h_G|$, and would correspond to homology localization (Corollary~\ref{cor:localization}) and measuring the state would reveal the planted clique.
\end{proof}

\begin{corollary}[Low-lying spectral density]\label{cor:llsd}
Fix threshold $\lambda\in(0,k)$ for any $k$ in the hard regime for planted clique. Under Assumption~\ref{ass:pc}, no
efficient algorithm estimates the eigenvalue count
$N_{<\lambda}(\Delta_{k-2}) := \#\{\text{eigenvalues of }\Delta_{k-2}\text{ in
}[0,\lambda)\}$ to additive error $<\tfrac12$ with success probability
$\ge\tfrac12+\delta$ for a constant $\delta>0$.
\end{corollary}
 
\begin{proof}
The spectrum of $\Delta_{k-2}$ on the planted skeleton is $\{0\}\cup \{k\}$, so for any $\lambda \in (0,k)$ the
count $N_{<\lambda}$ equals the multiplicity of $0$, namely $\betti_{k-2} = 1$; in the null case $\Delta_{k-2}$ acts on the zero space and
$N_{<\lambda} = 0$. Thus $N_{<\lambda}$ equals $\betti_{k-2}$ on both distributions, and
Theorem~\ref{thm:main} applies.
\end{proof}
 
\begin{remark}[Average-case analogue of the DQC1-hard LLSD problem]\label{rem:llsd}
Estimating the low-lying spectral density of a succinctly specified sparse
Hermitian operator is the DQC1-hard problem identified by Gyurik, Cade, and
Dunjko~\cite{GyurikCadeDunjko2022} as the computational core of quantum TDA:
LGZ-type algorithms estimate exactly this count for $\Delta_d$.
Corollary~\ref{cor:llsd} gives its \textit{average-case} analogue on the natural input
distribution.
\end{remark}
 
\begin{remark}
\label{rem:qma1_vsavg}
Deciding whether $H_d$ of a \emph{clique} complex vanishes is
\QMA$_1$-hard in the worst
case~\cite{CrichignoKohler2024,CadeCrichigno2024}. Corollary~5 is
not an average-case statement about that family. Rather,
Corollary~5 gives average-case hardness, on the distribution $\mathcal{D}_k$ of which clique homology is the
worst-case-hard subfamily. Whether homology vanishing of the clique complex itself is hard on average under any other natural distribution
remains open. 
\end{remark}

\subsection{Other implications}
\label{sec:otherimpl}
In this section we mention other implications of the findings presented in this work. 
\subsubsection{Restricted models (unconditional).}
Because the reduction is direct and the problem instance remains on the same graph, the unconditional
planted-clique lower bounds transfer to every estimation problem above with no
conjecture. Concretely: no low-degree polynomial test~\cite{KuniskyWeinBandeira2019,SchrammWein2022} and no statistical-query
algorithm~\cite{FeldmanGRVX2017} distinguishes the null and planted values of
$\betti_d$ in the PC window; and by the sum-of-squares lower bound for planted
clique~\cite{BarakHKKMP2019}, no polynomial-size, constant-degree
sum-of-squares certificate proves $H_d(\skel_d X(G))=0$ for
$G\sim\Gnp$, even though this vanishing holds w.h.p.

\subsubsection{Achieving   quantum advantage}
The ongoing search for quantum advantage in quantum Topological Data Analysis has resulted in the number of complexity results which rule out specific families of problems. In this work we show that both classical and quantumly the Betti number estimation problem as well as all the problems described in section \ref{sec:implications} are hard on average under the widely believed planted clique conjecture. We have identified a region where our reduction does not reach, leaving the average-case hardness question open. Moreover the identified region exhibits non-trivial homology, large spectral gaps, but still suffers from the state-preparation problem (P1), rendering existing quantum algorithms inefficient. However, for instances in this region if one identifies  perhaps instance specific efficient techniques for state-preparation then one might be capable of demonstrating instance- or family-specific quantum advantage.

\section{Open problems}\label{sec:open}

Our work opens up several directions for extending the
understanding of qTDA in the average-case setting. Our results also clarify what would be required to
obtain a convincing application of qTDA. A candidate family of instances to demonstrate quantum advantage should ideally satisfy the following conditions:
\begin{enumerate}[label=(\roman*),leftmargin=2em]
    \item the relevant simplices can be sampled approximately uniformly in
    polynomial time, and
    the corresponding simplex state can be prepared
    efficiently;
    \item a suitably normalized boundary operator or the combinatorial
    Laplacian can be implemented efficiently; 
    \item the smallest nonzero eigenvalue of the combinatorial Laplacian  
    satisfies
    \[
        \lambda_{\min}^{+}(\widetilde{\Delta}_d)
        \geq \frac{1}{\operatorname{poly}(n)};
    \]
    \item the required output is informative at inverse-polynomial additive
    precision.
\end{enumerate}
The most interesting spectral regime appears to be one in which the normalized
gap is inverse polynomial but tends to zero with the input size. At constant
gap, existing classical polynomial-filter algorithms may also run in
polynomial time. We conclude with a list of open problems that we hope will bring the complexity theory behind qTDA closer to practical applications.

\subsection{Extensions of the average-case hardness framework}

\begin{enumerate}[leftmargin=2em]

    \item \textbf{Average-case complexity under other natural distributions.}
    Determine the average-case complexity of Betti-number estimation under
    probability distributions other than the balanced planted-clique distribution
    considered here. 
   
    Even for pure Erd\H{o}s--R\'enyi flag complexes, it
    remains necessary to distinguish probabilistic concentration of the
    invariant from the computational complexity of estimating it at a
    specified precision.

    \item \textbf{Converse reductions and distributional completeness.}
    Determine whether planted clique can be recovered from a suitable
    Betti-number oracle in both directions, or whether Betti-number estimation
    is complete for a natural class of planted distributional problems. More
    generally, identify transformations between hidden subgraph problems and
    topological estimation problems that preserve the input distribution,
    success probability, and relevant access model.

    \item \textbf{Persistent and multiparameter extensions.}
    Extend our reduction to persistent homology and multiparameter
    persistence. A natural starting point is the skeletal filtration
    \[
        F_m(G)=\skel_m X(G).
    \]
    In the planted case, the boundary class created in
    $F_{k-2}(G)$ becomes trivial when the planted $(k-1)$-simplex is restored
    in $F_{k-1}(G)$. This suggests average-case hardness of deciding whether
    \[
        \ker\!\left(
        H_{k-2}(F_{k-2}(G))
        \longrightarrow
        H_{k-2}(F_{k-1}(G))
        \right)
    \]
    is nonzero. It would be useful to obtain analogous results for rank
    invariants, image persistence, cokernel persistence, and selected
    multiparameter Hilbert-function queries. This would complement existing
    quantum algorithms for persistent Betti numbers
    \cite{Hayakawa2022}.

    \item \textbf{Directed, path, and higher-order complexes.}
    Extend the average-case hardness results to directed flag complexes,
    path homology, hypergraph complexes, and other higher-order constructions.
    An undirected graph can be oriented according to a total vertex order,
    producing an acyclic directed graph whose directed flag complex is
    isomorphic to the original clique complex. This gives an immediate route
    to hardness results for directed flag homology. It remains open whether a
    comparable construction exists for grounded path homology and related
    directed persistence theories.

    \item \textbf{Average-case spectral and sampling parameters.}
    Determine the typical scaling of
    \[
        \frac{f_d}{\binom{n}{d+1}},
        \qquad
        \lambda_{\min}^{+}(\widetilde{\Delta}_d),
        \qquad
        \frac{\beta_d}{f_d},
    \]
    under natural probability distributions. The analysis should distinguish
    membership access, local sparse-row access, approximate simplex sampling,
    coherent state preparation, and block encoding. Uniform lower bounds on the
    normalized Hodge gap for growing dimension would be particularly valuable.

    \item \textbf{Succinct boundary-matrix complexity.}
    Formalize the complexity of rank, nullity, and low-energy spectral
    estimation for exponentially large boundary matrices given through a
    compact graph or hypergraph description. The resulting problem is
    different from rank computation for an explicitly listed sparse matrix,
    which is polynomial-time solvable. It would be useful to identify complete
    problems for natural succinct-access models and to determine which
    classical restricted-model lower bounds transfer to them.

\end{enumerate}

\subsection{Biological applications}

\begin{enumerate}[leftmargin=2em]

    \item \textbf{Online and zigzag qTDA for evolving tissues.}
    Develop quantum algorithms for time-dependent sequences
    \[
        K_0 \longleftrightarrow K_1
        \longleftrightarrow \cdots
        \longleftrightarrow K_T,
    \]
    in which simplices can both appear and disappear. An important question is
    whether a simplex state or harmonic state prepared at time $t$ can be used
    as a warm start at time $t+1$, leading to an amortized advantage over
    repeated classical matrix reduction. This has a number of practical biological applications which include tumour evolution, immune-cell migration, and vascular remodelling.

    \item \textbf{Directed topology of biological flow and signalling.}
    Develop qTDA methods for directed biological networks, including
    signalling pathways, gene-regulatory networks, metabolic networks,
    vascular flow, neuronal circuits, and polarized molecular interactions.

    \item \textbf{Higher-order topology in multi-omics and cellular
    interactions.}
    Construct biologically motivated simplicial or hypergraph models in which
    a higher-dimensional simplex represents a genuine joint interaction
    rather than only the clique completion of pairwise edges. Examples include multi-cell signalling events, protein complexes,
    gene-regulatory modules, microbiome communities and beyond. One should determine whether higher-order topological features improve prediction beyond graph-based models and whether the corresponding simplices can be sampled without enumerating the
    full complex.

\end{enumerate}

\subsection{Quantum advantage}

\begin{enumerate}[leftmargin=2em]

    \item \textbf{Common access models for classical and quantum algorithms.}
    Compare quantum algorithms with the corresponding classical algorithms under
    the same input assumptions. The comparison should include exact sparse
    boundary reduction, path-integral estimators, Krylov and Chebyshev methods,
    stochastic Lanczos quadrature, and application-specific sampling
    algorithms \cite{Apers2022,Berry2024}.
    
    \item \textbf{Amortized advantage.}
    Many applications rely on performing the same topological computation over
    many scales, time points, parameter values, candidate molecules, or
    bootstrap samples. Determine whether quantum states, block encodings, or
    spectral estimates can be reused across related instances. An amortized
    advantage over a long sequence of correlated queries may be achievable
    even when no advantage is present for a single isolated complex.

\end{enumerate}

\subsection{Cryptography}

\begin{enumerate}[leftmargin=2em]

    \item \textbf{Betti-number estimation as a proof-of-work problem.}
    Determine whether the average-case hardness of Betti-number estimation can
    be combined with an efficient verification mechanism. 
    
 \item \textbf{Cryptographic primitives based on topological problems.}
    Determine whether average-case hardness of Betti-number estimation,
    homology vanishing, homology localization, or related succinctly
    represented linear-algebra problems can be used to construct standard
    cryptographic primitives such as one-way functions.

\end{enumerate}

\section*{Acknowledgements}
Strelchuk acknowledges support from the Wellcome Leap as part of the Q4Bio Program and the Royal Society University Research Fellowship. Subramanian acknowledges support from the Royal Society through a University Research Fellowship.


\bibliographystyle{alpha}
\bibliography{ref}

\end{document}